\documentclass[10pt]{amsart}

\usepackage{graphicx} % a package for including diagrams
\usepackage{setspace} % provides the doublespace command below
\usepackage{amsfonts} % another must have AMS fonts
\usepackage[margin=1in]{geometry}
\usepackage{amsthm}
\usepackage{bbm}
\usepackage{amsfonts,amssymb,latexsym,amsmath,enumerate}
\usepackage{color}
\usepackage{soul}
\newtheorem{theorem}{Theorem}[section]
\newtheorem{definition}[theorem]{Definition}

\newtheorem{lemma}[theorem]{Lemma}
\newtheorem{remark}[theorem]{Remark}

\newtheorem{algorithm}[theorem]{Algorithm}
\newtheorem{example}[theorem]{Example}
\newcommand{\dotProduct}[2]{\ensuremath{\left\langle #1 , #2\right\rangle}}
\numberwithin{equation}{section} 

\begin{document}
\title[$k$-angle tight frames]{Construction of $\boldsymbol{k}$-angle tight frames}
\author{Somantika Datta and Jesse Oldroyd}
\address{Department of Mathematics, University of Idaho, Moscow, ID 83844-1103, USA}
\email{sdatta@uidaho.edu, jesseo@uidaho.edu}
\thanks{This material is based upon work supported by the National Science Foundation under Award No. CCF-1422252}
\date{\today}
\doublespacing
\begin{abstract}
%An explicit construction of real or complex equiangular tight frames of $d+1$ vectors in a $d$ dimensional space is presented.
%With the goal of generalizing equiangular tight frames, the paper presents a construction of tight frames whose corresponding Gram matrix has at most $k$-distinct values in moduli.
Frames have become standard tools in signal processing due to their robustness to transmission errors and their resilience to noise. Equiangular tight frames (ETFs) are particularly useful and have been shown to be optimal for transmission under a certain number of erasures. Unfortunately, ETFs do not exist in many cases and are hard to construct when they do exist. However, it is known that an ETF of $d+1$ vectors in a $d$ dimensional space always exists. This paper gives an explicit construction of ETFs of $d+1$ vectors in a $d$ dimensional space. This construction works for both real and complex cases and is simpler than existing methods. The absence of ETFs of arbitrary sizes in a given space leads to generalizations of ETFs. One way to do so is to consider tight frames where the set of (acute) angles between pairs of vectors has $k$ distinct values. This paper presents a construction of tight frames such that for a given value of $k$, the angles between pairs of vectors take at most $k$ distinct values. These tight frames can be related to regular graphs and association schemes.
\end{abstract}
\maketitle


%\begin{keyword}
%%% keywords here, in the form: keyword \sep keyword
%$k$-angle tight frames \sep equiangular frames \sep signature matrix \sep tight frames \sep Welch bound
%%% PACS codes here, in the form: \PACS code \sep code
%
%%% MSC codes here, in the form:
%\MSC 42C15 \sep 94Axx
%%% or \MSC[2008] code \sep code (2000 is the default)
%\end{keyword}
\section{Introduction}
\label{Intro}
\subsection{Background and motivation}
The maximum cross correlation between pairs in a set of $N$ unit vectors $\{f_1, \ldots, f_N\}$ in $\mathbb{C}^d$ is bounded below by the Welch bound \cite{W1}:
\begin{equation} \label{WelchInEq}
\max_{i \neq j} |\langle f_i, f_j \rangle| \geq \sqrt{\frac{N-d}{d(N-1)}},\quad N\geq d.
\end{equation}
It is well known that equality is attained in (\ref{WelchInEq}) when the set $\{f_1, \ldots, f_N\}$ is an equiangular tight frame (ETF) \cite{SH1, DattaHC}. Sets that attain the Welch bound arise in many different areas as in communications, quantum information processing, and coding theory \cite{W1,Ren04, Sco06, KR1, RoyScott07, Hoggar82}. Consequently, the problem of constructing ETFs and determining conditions under which they exist has gained substantial attention \cite{SH1, Holmes04, Tropp05, Sustik07, Bodmann09, Bodmann2010, Waldron09,Fickus2012}.  For an ETF, the associated Gram matrix can be written as
\begin{equation}\label{Gram matrix}
G = I + \alpha Q
\end{equation}
where $I$ is the identity matrix, $\alpha$ is the Welch bound $\sqrt{\frac{N-d}{d(N-1)}},$ and $Q$ is a Hermitian matrix with zeros along the diagonal and unimodular entries elsewhere.
% Note that this is the case when the vectors are unit-normed. For convenience, this will be the assumption here. Simple modifications can be made for other situations. For example, if $\{f_i\}_{i=1}^N$ is an equi-normed Parseval frame, then $\|f_i\|^2 = \frac dN,$ $i = 1, \ldots, N,$ and
% $$
% G = \frac dN I + \alpha Q
% $$
% where $\alpha = \frac dN \sqrt{\frac{N-d}{d(N-1)}}.$
% For a unit-normed tight frame,  frame bound is $N/d.$
This means that the Gram matrix has two distinct eigenvalues: zero and $\frac Nd$ with multiplicities $N-d$ and $d,$ respectively. From (\ref{Gram matrix}), this implies that the two distinct eigenvalues of the matrix $Q$ are
\begin{eqnarray} \label{EvaluesQ}
\lambda_1 &=& - \frac{1}{\alpha} = -\sqrt{\frac{d(N-1)}{N-d}},  \nonumber \\
\lambda_2 &=& \frac{N-d}{d\alpha} = \sqrt{\frac{(N-1)(N-d)}{d}}
\end{eqnarray}
with multiplicities $N-d$ and $d,$ respectively.
The matrix $Q$ in (\ref{Gram matrix}) is known as the \emph{signature matrix} of the ETF \cite{Holmes04}.  The problem of constructing an ETF thus reduces to the task of constructing a signature matrix $Q$ with eigenvalues given by  (\ref{EvaluesQ}). Conditions on $Q$ for being the associated signature matrix of an ETF have been discussed in \cite{SH1,Holmes04, Sustik07, Bodmann09, Bodmann2010} among others. A graph theoretic approach to constructing ETFs has been studied in \cite{Waldron09}. A correspondence recently discovered by Fickus et al.~\cite{Fickus2012} uses Steiner systems to directly construct the frame vectors of certain ETFs, bypassing the common technique of constructing a suitable Gram matrix or signature matrix. This approach lets one construct highly redundant sparse ETFs. However, in the real case this approach can give rise to ETFs only if a real Hadamard matrix of a certain size exists.

Despite the desirability and importance of ETFs, these cannot exist for many pairs $(N, d).$ When the Hilbert space is $\mathbb{R}^d,$ the maximum number of equiangular lines is bounded by $\frac{d(d+1)}{2}$ and for $\mathbb{C}^d$ the bound is $d^2$ \cite{Lemmens73, DGS75}. Even when these restrictions hold, ETFs are very hard to construct and do not exist for many pairs $(N,d)$ \cite{Sustik07}. This leads one to generalizing the notion of an ETF.
For a real ETF, the off-diagonal entries of the Gram matrix are either $\alpha$ or $-\alpha$, where $\alpha$ is the Welch bound. In other words, the off-diagonal entries of the Gram matrix all have modulus equal to $\alpha.$ Generalizing this notion, a tight frame whose associated Gram matrix has ones along the diagonal and off-diagonal entries with $k$ distinct \textit{moduli} will be called a \emph{$k$-angle tight frame}. Under this definition, ETFs are viewed as $1$-angle tight frames.\footnote{It is to be noted that often in the literature, a unit-normed tight frame is called a \emph{two-distance tight frame} \cite{Barg2014, Larman77} if the off-diagonal entries of the associated Gram matrix take on either of two values $a$ and $b.$ In that case, real ETFs are thought of as two-distance tight frames instead of $1$-angle tight frames, as done here.}  Besides generalizing the notion of an ETF, $k$-angle tight frames prove to be important also due to their connection to graphs and association schemes as discussed in Section~\ref{k_angle_frames_graphs}. It is worth mentioning here that sets of vectors such that the absolute value of the inner product between distinct vectors takes $k$ distinct values has been mentioned in~\cite{Brouwer2011}, and upper bounds on the size of such sets form the content of the fundamental work done in~\cite{DGS75,DGS1}. However, explicit constructions of such sets for arbitrary $k$ do not seem to exist in the literature.

With the above motivation in mind, the main contribution of the work presented here involves the construction of $\hat{k}$-angle tight frames with $\hat{k}$ being less than or equal to some given positive integer $k$ (see Theorem~\ref{theorem::binomialConstruction}). A straightforward construction of ETFs of $d+1$ vectors in $\mathbb{R}^{d}$ or $\mathbb{C}^{d}$ is also presented. It is known that in this case the existence of ETFs is guaranteed.
%Construction of ETFs  of $d+1$ vectors in $\mathbb{R}^d$ or $\mathbb{C}^d$ has been discussed in \cite{Goyal_2001, SH1}.
A nice construction suggested in \cite{Goyal_2001} leads to ETFs for $\mathbb{C}^d$, whereas the construction presented here applies to both $\mathbb{C}^d$ and $\mathbb{R}^d$, and is simpler. Several constructions of $2$-angle tight frames in $\mathbb{C}^d$ or $\mathbb{R}^d$ are also discussed in this work and are connected to mathematical objects called mutually unbiased bases (MUBs).
%Given $1< M < N,$ this paper also gives a new way of constructing tight frames of $N$ vectors in $\mathbb{R}^{N-M}$ or $\mathbb{C}^{N-M}$ from an orthogonal set of $M$ vectors in  $\mathbb{R}^N$ or $\mathbb{C}^N.$  For such tight frames, upper bounds are obtained on $k$ by giving upper bounds on the number of unique values taken by the entries of the related Gram matrix.	

% Before beginning the mathematical parts of the paper some notation and other preliminaries are introduced below.
%-------------
\subsection{Notation and preliminaries}
Given a set $\{f_1, \ldots, f_N\}$ of vectors in $\mathbb{R}^{d}$ or $\mathbb{C}^d,$ let $F$ be the matrix whose columns are the vectors $f_1, \ldots, f_N.$ For a \emph{tight frame} the $d \times d$ matrix $F F^*$ is a multiple of the identity. The matrix $F^*F$ is the Gram matrix of the set $\{f_1, \ldots, f_N\}$ and has the same non-zero eigenvalues as those of $FF^*.$ The entries of the Gram matrix are the inner products of the vectors $\{f_1, \ldots, f_N\}.$
By  an \emph{equiangular tight frame} (ETF) is meant a tight frame $\{f_1, \ldots, f_N\}$ for a $d$ dimensional space $\mathcal{H}$ such that the frame bound is $\frac Nd,$ $\|f_i\| = 1,$ for $i = 1, \ldots, N,$ and $|\langle f_i, f_j \rangle| = \alpha, $ $1\leq i\neq j \leq N.$ Here $\alpha $ is the Welch bound  given in (\ref{WelchInEq}). Throughout, $\mathcal{H}$ will be either $\mathbb{C}^d$ or $\mathbb{R}^d$.\footnote{The results can be easily generalized to any $d$-dimensional Hilbert space $\mathcal{H}$ since $\mathcal{H}$ would be isomorphic to $\mathbb{R}^d$ or $\mathbb{C}^d.$} A frame of $N$ vectors in $\mathbb{R}^{d}$ (respectively, $\mathbb{C}^{d}$) will be referred to as an $(N,d)$ real (respectively, complex) frame.
%A real or complex frame of $N$ vectors for $\mathbb{C}^d$ or $\mathbb{R}^d$ will be referred to as an $(N, d)$ real or complex frame.
When $\mathcal{H}$ is not specified, the frame will  be called an $(N, d)$ frame.
%In general, any Hermitian matrix $Q$ with zeros on the diagonal and unimodular entries elsewhere will be referred to as a Seidel matrix.
If $Q$ corresponds to an ETF in the sense of (\ref{Gram matrix}) then it will be called a \textit{signature matrix}.
%A Hermitian matrix $Q$ with zeros on the diagonal will be referred to as a \textit{generalized Seidel matrix}.

%The entry in the $m^{\text{th}}$ row and $n^{\text{th}}$ column of an arbitrary matrix $A$ will be denoted by $[A]_{mn}$. The $N\times N$ identity matrix will be denoted by $I_{N}$ or just $I$ when there is no confusion. By a \textit{circulant matrix} is meant a square matrix $C$ of the form
%\[
%	C =
%	\begin{bmatrix}
%		c_{0} & c_{N-1} & c_{N-2} & \cdots & c_{1} \\
%		c_{1} & c_{0} & c_{N-1} & \cdots & c_{2} \\
%		c_{2} & c_{1} & c_{0} & \cdots & c_{3} \\
%		\vdots & \vdots & \vdots & \ddots & \vdots \\
%		c_{N-1} & c_{N-2} & c_{N-3} & \cdots & c_{0}		
%	\end{bmatrix}.
%\]
%-------------
\subsection{Outline} \label{outline}
The paper is divided as follows. Section \ref{k_angle_frames_graphs} motivates the study of $k$-angle tight frames by discussing their relation to concepts in graph theory and coding theory.
Given a positive integer $d,$ a simple and straightforward method of constructing a $(d+1, d)$ real or complex equiangular tight frame is discussed in Section \ref{ComplexETF}.
In Section \ref{k_angle_tight_frames}, the main result on $k$-angle tight frames is given in Theorem~\ref{theorem::binomialConstruction}. For a given $k,$ Theorem~\ref{theorem::binomialConstruction} gives a construction of tight frames such that there are at most $k$ distinct angles between pairs of vectors.
% Section \ref{Discuss} discusses some future work including an idea of how one can think of constructing approximate equiangular tight frames in situations where equiangular tight frames do not exist or are hard to construct.
\label{Outline}
%-----------------------------------------
\section{$k$-angle tight frames, regular graphs, and association schemes}\label{k_angle_frames_graphs}
As already mentioned in Section~\ref{Intro} above, $k$-angle tight frames can be connected to mathematical objects arising in graph theory and coding theory such as regular graphs and association schemes.
The connection of ETFs to graphs is as follows \cite{SH1}. Suppose that the Gram matrix $G$ associated with an ETF has ones along the diagonal and $\pm \alpha$ elsewhere. Then
$$Q = \frac{1}{\alpha}(G - I)$$
is the Seidel adjacency matrix of a regular two-graph \cite{Brouwer2011,Seidel73}.
% The matrix $Q$ is sometimes also called a \emph{Seidel matrix}.
Barg et al.~\cite{Barg2014} have shown a correspondence between non-equiangular $2$-angle tight frames and strongly regular graphs. In the case of $3$-angle tight frames an analogous connection may be drawn to regular graphs. In particular, let $G = I+c_{1}Q_{1}+c_{2}Q_{2}+c_{3}Q_{3}$ be the Gram matrix of a $3$-angle tight frame where $c_{i}\neq\pm c_{j}$ for $i\neq j$, $c_{i}\neq0$ and $Q_{i}$ is a zero diagonal symmetric binary matrix for $i=1,2,3$. Then $Q_{i}$ is the adjacency matrix for a regular graph for $i=1,2,3$ if and only if $u = [1 \ldots 1]^{\mathrm{T}}\in\mathbb{R}^{N}$ is an eigenvector of $G$. The details will form part of  a separate paper.

Certain $k$-angle tight frames also provide examples of association schemes~\cite{Brouwer2011}. If $G = I+c_{1}Q_{1}+\cdots+c_{k}Q_{k}$ is the Gram matrix of a $k$-angle tight frame, where $Q_{i}$ is a zero diagonal symmetric binary matrix for $1\leq i\leq k$, then $\{I,Q_{1},\ldots,Q_{k}\}$ forms an association scheme if $Q_{i}Q_{j} = Q_{j}Q_{i}$ for $1\leq i,j\leq k$.

Further, $k$-angle tight frames are specific examples of what Delsarte et al.~\cite{DGS75} refer to as \textit{$A$-sets}. For a given finite dimensional Hilbert space, upper bounds on the size of an $A$-set, and therefore on the number of vectors in a $k$-angle tight frame, are given in~\cite{DGS75,Brouwer2011}.

%-----------------------------------------
\section{Construction of $(d+1, d)$ equiangular tight frames}
\label{ComplexETF}

Goyal and Kova\v{c}evi\'{c}~\cite{Goyal_2001} have previously given an elegant characterization of $(d+1,d)$ complex ETFs in terms of harmonic tight frames. Although this allows finding frame expansions by using Fast Fourier Transform algorithms, computing the frame vectors themselves requires a series of $d$ trigonometric evaluations and $d$ non-trivial scalar multiplications. If a trigonometric evaluation is considered as a single operation, then using harmonic tight frames to get a $(d+1,d)$ ETF requires $O(d^{2})$ operations for each vector. Theorem~\ref{theorem::TrivialSignatureMatrix} below takes a different approach by characterizing the signature matrices of real as well as complex $(d+1,d)$ ETFs, whereas results in~\cite{Goyal_2001} only give complex ETFs. A benefit of this result is that it gives a method to compute the vectors of a $(d+1,d)$ ETF such that each frame vector may be computed using only $O(d)$ operations (see Remark~\ref{remark::operation_count}).

% Every Seidel matrix with precisely two eigenvalues is the signature matrix of an ETF. Therefore methods which determine such Seidel matrices automatically provide ETFs as well. By (\ref{Gram matrix}), the Gram matrix of the ETF can be obtained once the signature matrix is known, and the ETF can then be constructed from the Gram matrix, see Algorithm \ref{ConstructETFfromQ}. This section outlines two different approaches to characterizing such matrices: the first is a method to construct all possible signature matrices for $(d+1,d)$ ETFs and the second is a characterization of all possible signature matrices of $(2d,d)$ ETFs.\footnote{In the context of this paper, this is the same as the construction of $1$-angle $(d+1, d)$ and $(2d, d)$ tight frames.}
%
%\subsection{Construction of complex equiangular tight frames of size $N$ in $\mathbb{C}^{N-1}$}
%

Theorem \ref{theorem::TrivialSignatureMatrix} below is a complete, constructive characterization of signature matrices of $(d+1,d)$ ETFs. Due to (\ref{Gram matrix}), the Gram matrix of a $(d+1,d)$ ETF is
\begin{equation}\label{Gram matrix:NN-1ETF}
	G = I _{d+1}+ \frac{1}{d}Q
\end{equation}
with eigenvalues $0$ and $\frac{d+1}{d}$ where $I_{d+1}$ denotes the $(d+1)\times(d+1)$ identity matrix.
It follows from work in \cite{Holmes04} that being a $(d+1,d)$ ETF is equivalent to the signature matrix $Q$ satisfying
\begin{equation}\label{equation::HolmesSignatureMatrixIdentity}
	Q^{2} = (\lambda_{1}+\lambda_{2})Q-\lambda_{1}\lambda_{2}I_{d+1}
\end{equation}
where $\lambda_{1} = -d$ and $\lambda_{2}=1$ are the eigenvalues of $Q$ in this case. This fact will be used in the proof of Theorem \ref{theorem::TrivialSignatureMatrix}. Even though the construction in Theorem \ref{theorem::TrivialSignatureMatrix} below is done for complex ETFs, the exact same construction gives real $(d+1,d)$ ETFs as well.

\begin{theorem}\label{theorem::TrivialSignatureMatrix}
	Let $Q$ be a $(d+1)\times (d+1)$ matrix with complex entries. Then $Q$ is a signature matrix for a $(d+1,d)$ complex ETF if and only if $Q = I_{d+1}-xx^{*}$ for some $x\in\mathbb{C}^{d+1}$ with unimodular entries.
\end{theorem}
\begin{proof}
	Let $x\in\mathbb{C}^{d+1}$ have unimodular entries and let $Q = I_{d+1}-xx^{*}$. By computation, and using the fact that $\|x\|^{2} = d+1$, it follows that
	\begin{align*}
		Q^{2}
		&= I_{d+1}-2xx^{*}+(d+1)xx^{*} \\
		&= Q + dxx^{*} \\
		&= Q + dxx^{*}+d I_{d+1}-d I_{d+1} \\
		&= Q - d Q+d I_{d+1} \\
		&= (1-d)Q-(-d)I_{d+1} \\
		&= (\lambda_{1}+\lambda_{2})Q - \lambda_{1}\lambda_{2}I_{d+1}.
	\end{align*}
	This shows that every matrix of the form $Q = I_{d+1}-xx^{*},$ for $x\in\mathbb{C}^{d+1}$ with unimodular entries, satisfies (\ref{equation::HolmesSignatureMatrixIdentity}) and is therefore the signature matrix for a $(d+1,d)$ ETF.

	Now let $Q$ be a signature matrix for a complex $(d+1,d)$ ETF. By~(\ref{EvaluesQ}), $Q$ is a Hermitian matrix with eigenvalues $\lambda_{1}=-d$ and $\lambda_{2}=1.$ Note that the multiplicities of $\lambda_1 = -d$ and $\lambda_2 = 1$ are $1$ and $d,$ respectively. Let $x$ be an eigenvector associated with $\lambda_{1}=-d$ and satisfying $\|x\|^{2} = d+1$. Since $Q$ is Hermitian there exists an orthogonal basis for $\mathbb{C}^{d+1}$ of eigenvectors of $Q$, say $\{x,y_{1},\ldots,y_{d}\},$ where $y_{j}$, $1\leq j\leq d$, are eigenvectors for the eigenvalue $1$. Let ${z}\in\mathbb{C}^{d+1}.$ Then ${z}$ can be written as  $${z} = \sum_{j=1}^{d}c_{j}y_{j}+c_{d+1}x$$ and
	\begin{align*}
		Q{z}
		&= \sum_{j=1}^{d}c_{j}Qy_{j}+c_{d+1}Qx = \sum_{j=1}^{d}c_{j}y_{j} - c_{d+1}dx\\
		&= {z}-(d+1)c_{d+1}x.
	\end{align*}
	On the other hand, a similar calculation using the orthogonality of the set $\{x,y_{1},\ldots,y_{d}\}$ and the fact that $\|x\|^{2} = d+1,$ yields
	\[
		(I_{d+1}-xx^{*}){z} = {z}-(d+1)c_{d+1}x.
	\]
	Since ${z}$ was arbitrary, it follows that $Q = I_{d+1}-xx^{*}$. To see that $x = (x_{j})_{1\leq j\leq d+1}$ has unimodular entries, note that since $Q$ has zeros along the diagonal, the equality $Q = I_{d+1}-xx^{*}$ forces $x_{j}\overline{x}_{j} = 1$ for $1\leq j\leq d+1$.
\end{proof}
%
% \begin{remark}
% \rm
% Even though the construction in Theorem \ref{theorem::TrivialSignatureMatrix} is done for complex ETFs, the exact same construction gives real $(d+1,d)$ ETFs as well. The vector $x$ is then taken so that each entry is $\pm 1.$
% \end{remark}
%
\begin{remark}\label{Rem:EVectorsofQ&G}
\rm
Any vector $x \in \mathbb{C}^{d+1}$ with unimodular entries is an eigenvector of $Q = I_{d+1}-xx^{*}$ corresponding to the eigenvalue $-d.$ Further, the signature matrix $Q$ and the corresponding Gram matrix $G$ have the same eigenvectors. From the proof of Theorem \ref{theorem::TrivialSignatureMatrix}, the set $\{x, y_1, \ldots, y_{d}\}$ is also a set of orthogonal eigenvectors of $G.$ The eigenvalue of $G$ for the eigenvector $x$ is zero.
\end{remark}
Algorithm~\ref{ConstructETFfromQ} below outlines how Theorem~\ref{theorem::TrivialSignatureMatrix} may be used to construct a $(d+1, d)$ ETF. Recall that for a $(d+1,d)$ ETF, the Welch bound $\alpha$ is $\frac{1}{d}$.
\begin{algorithm}\label{ConstructETFfromQ}
$\,$
\begin{enumerate}[Step 1:]
	\item Choose a vector $x$ in $\mathbb{R}^{d+1}$ or $\mathbb{C}^{d+1}$ with unimodular entries, and
	construct the signature matrix $Q$ from $x$ as described in Theorem~\ref{theorem::TrivialSignatureMatrix}.

	\item Construct the corresponding Gram matrix $G=I+\frac{1}{d} Q$.
	%For a $(d+1,d)$ ETF, $\alpha$ is $\frac{1}{d}.$
%	\[
%		\alpha = \sqrt{\frac{N-d}{d(N-1)}} = \frac{1}{N-1}.
%	\]

	\item Diagonalize $G$ into $G=UDU^{*}$, where $U$ is a unitary matrix of eigenvectors
%\footnote{Note that $G$ and $Q$ have the same eigenvectors. From the proof of Theorem \ref{theorem::TrivialSignatureMatrix}, the set $\{x, y_1, \ldots, y_{d}\}$ is a set of orthogonal eigenvectors of $G.$}
of $G$ and $D$ is the diagonal matrix of corresponding eigenvalues arranged in descending order. For a $(d+1,d)$ ETF:
	\[
		D = \mathrm{diag}\left(\left\{\underbrace{\frac{d+1}{d},\frac{d+1}{d},\ldots,\frac{d+1}{d}}_{d\text{ times}},0\right\}\right).
	\]

	\item Obtain the frame vectors from the rows of the matrix $U\sqrt{D}$, where $\sqrt{D}$ denotes the diagonal matrix whose entries are the positive square roots of corresponding entries of $D$.
\end{enumerate}
\end{algorithm}

\begin{example}[A real $(6,5)$ ETF]
	\rm
	Let the vector $x\in \mathbb{R}^{6}$ be $ [1, 1, -1, 1, -1, 1]^{\textrm{T}}.$ Since $\alpha = \frac{1}{5}$, the Gram matrix is
%	\[
%		x =
%		\begin{bmatrix}
%			1\\
%			1\\
%			-1\\
%			1\\
%			-1\\
%			1	
%		\end{bmatrix}.
%	\]
	% By Theorem~\ref{theorem::TrivialSignatureMatrix}, the matrix $Q$ determined by
	% \begin{align*}
	% 	Q
	% 	&= I_{6} - xx^{\mathrm{T}} \\
	% 	&=\left[
	% 	\begin{array}{rrrrrr}
	% 		0 & -1 & 1 & -1 & 1 & -1 \\
	% 		-1 & 0 & 1 & -1 & 1 & -1 \\
	% 		1 & 1 & 0 & 1 & -1 & 1 \\
	% 		-1 & -1 & 1 & 0 & 1 & -1 \\
	% 		1 & 1 & -1 & 1 & 0 & 1 \\
	% 		-1 & -1 & 1 & -1 & 1 & 0
	% 	\end{array}\right]
	% \end{align*}
	% must be the signature matrix for a real $(6,5)$ ETF.

	% The next step is to construct the Gram matrix.
	\begin{align*}
		G
		&= I_{6}+\frac{1}{5}Q
		= \left[
		\begin{array}{rrrrrr}
			1 & -\frac{1}{5} & \frac{1}{5} & -\frac{1}{5} & \frac{1}{5} & -\frac{1}{5} \\
			-\frac{1}{5} & 1 & \frac{1}{5} & -\frac{1}{5} & \frac{1}{5} & -\frac{1}{5} \\
			\frac{1}{5} & \frac{1}{5} & 1 & \frac{1}{5} & -\frac{1}{5} & \frac{1}{5} \\
			-\frac{1}{5} & -\frac{1}{5} & \frac{1}{5} & 1 & \frac{1}{5} & -\frac{1}{5} \\
			\frac{1}{5} & \frac{1}{5} & -\frac{1}{5} & \frac{1}{5} & 1 & \frac{1}{5} \\
			-\frac{1}{5} & -\frac{1}{5} & \frac{1}{5} & -\frac{1}{5} & \frac{1}{5} & 1
		\end{array}
		\right].
	\end{align*}
	Since the last column of $U\sqrt{D}$ is $0$, a real $(6,5)$ ETF is then given by the rows of the matrix
	% \[
	% 	U\sqrt{D}
	% 	= \sqrt{\frac{6}{5}}
	% 	\begin{bmatrix}
	% 		u_{1} & u_{2} & u_{3} & u_{4} & u_{5} & {0}
	% 	\end{bmatrix}
	% \]
	% after truncation of the last column. In particular, the rows of
	\begin{align*}
		\sqrt{\frac{6}{5}}
		\begin{bmatrix}
			u_{1} & u_{2} & u_{3} & u_{4} & u_{5}
		\end{bmatrix}
		&= \left[
			\begin{array}{rrrrrr}
				\sqrt{\frac{3}{5}} & \frac{1}{2}\sqrt{\frac{4}{5}} & \frac{1}{3}\sqrt{\frac{9}{10}} & \frac{1}{4}\sqrt{\frac{24}{25}} & \frac{1}{5} \\
				-\sqrt{\frac{3}{5}} & \frac{1}{2}\sqrt{\frac{4}{5}} & \frac{1}{3}\sqrt{\frac{9}{10}} & \frac{1}{4}\sqrt{\frac{24}{25}} & \frac{1}{5} \\
				0 & \sqrt{\frac{4}{5}} & -\frac{1}{3}\sqrt{\frac{9}{10}} & -\frac{1}{4}\sqrt{\frac{24}{25}} & -\frac{1}{5} \\
				0 & 0 & -\sqrt{\frac{9}{10}} & \frac{1}{4}\sqrt{\frac{24}{25}} & \frac{1}{5} \\
				0 & 0 & 0 & \sqrt{\frac{24}{25}} & -\frac{1}{5} \\
				0 & 0 & 0 & 0 & -1 \\
			\end{array}
		\right].
	\end{align*}
	% form a real $(6,5)$ ETF.
\end{example}
\begin{example}[A complex $(4,3)$ ETF]
	\rm
	%Algorithm \ref{ConstructETFfromQ} next is applied to the creation of a complex $(4,3)$ ETF.
	Let $x\in\mathbb{C}^{4}$ be given by $x = [1, i, -1, -i]^{\textrm{T}}.$
%	\[	x =
%		\begin{bmatrix}
%			1 \\
%			i \\
%			-1 \\
%			-i	
%		\end{bmatrix}.
%	\]
	% Then the matrix
	% \begin{align*}
	% 	Q
	% 	&= I_{4}-xx^{*} \\
	% 	&= \left[
	% 		\begin{array}{rrrr}
	% 			0 & i & 1 & -i \\
	% 			-i & 0 & i & 1 \\
	% 			1 & -i & 0 & i \\
	% 			i & 1 & -i & 0
	% 		\end{array}
	% 	\right]
	% \end{align*}
	% is the signature matrix for a complex $(4,3)$ ETF.
	The Gram matrix of the ETF is
	\begin{align*}
		G
		% &= I_{4}+\frac{1}{3}Q \\
		&= \left[
			\begin{array}{rrrr}
				1 & \frac{i}{3} & \frac{1}{3} & -\frac{i}{3} \\
				-\frac{i}{3} & 1 & \frac{i}{3} & \frac{1}{3} \\
				\frac{1}{3} & -\frac{i}{3} & 1 & \frac{i}{3} \\
				\frac{i}{3} & \frac{1}{3} & -\frac{i}{3} & 1
			\end{array}
		\right].
	\end{align*}
	% $G$ may be diagonalized by finding three mutually orthogonal vectors that are also orthogonal to $x$ and then normalizing. In particular, $G = UDU^{*}$ where
	% \begin{align*}
	% 	U
	% 	&=
	% 	\begin{bmatrix}
	% 		u_{1} & u_{2} & u_{3} & u_{4}
	% 	\end{bmatrix} \\
	% 	&= \left[
	% 		\begin{array}{rrrr}
	% 			-\frac{i}{\sqrt{2}} & -\frac{1}{2}\sqrt{\frac{2}{3}} & \frac{i}{3}\sqrt{\frac{3}{4}} & \frac{1}{2} \\
	% 			-\frac{1}{\sqrt{2}} & -\frac{i}{2}\sqrt{\frac{2}{3}} & -\frac{1}{3}\sqrt{\frac{3}{4}} & \frac{i}{2} \\
	% 			0 & -\sqrt{\frac{2}{3}} & -\frac{i}{3}\sqrt{\frac{3}{4}} & -\frac{1}{2} \\
	% 			0 & 0 & -\sqrt{\frac{3}{4}} & -\frac{i}{2}
	% 		\end{array}
	% 	\right]
	% \end{align*}
	% and
	% \[
	% 	D = \left[
	% 		\begin{array}{rrrr}
	% 			\frac{4}{3} & 0 & 0 & 0 \\
	% 			0 & \frac{4}{3} & 0 & 0 \\
	% 			0 & 0 & \frac{4}{3} & 0 \\
	% 			0 & 0 & 0 & 0
	% 		\end{array}
	% 	\right].
	% \]
	The row vectors of
	\[
		\sqrt{\frac{4}{3}}
		\begin{bmatrix}
			u_{1} & u_{2} & u_{3}
		\end{bmatrix}
		= \left[
		\begin{array}{rrrr}
			-i\sqrt{\frac{2}{3}} & -\frac{\sqrt{2}}{3} & \frac{i}{3} \\
			-\sqrt{\frac{2}{3}} & -i\frac{\sqrt{2}}{3} & -\frac{1}{3} \\
			0 & -\frac{2\sqrt{2}}{3} & -\frac{i}{3}\\
			0 & 0 & -1
		\end{array}
		\right]
	\]
	form a complex $(4,3)$ ETF.
\end{example}

\begin{remark}\label{remark::operation_count}
\rm
	It can be checked that the vectors
$y_{j} = [x_{1}/j, x_{2}/j, \ldots, x_{j-1}/j, x_{j}/j, -x_{j+1}, 0, \ldots, 0]^{\textrm{T}}$
	%\[
%		y_{i} =
%		\begin{bmatrix}
%			x_{1}/i \\ x_{2}/i \\ \vdots \\ x_{i-1}/i \\ x_{i}/i \\ -x_{i+1} \\ 0 \\ \vdots \\ 0
%		\end{bmatrix}
%	\]
	form an orthogonal basis of eigenvectors for the Gram matrix of the real $(d+1,d)$ ETF with signature matrix $Q = I-xx^{\textrm{T}}$, where $x = [x_{j}]_{1\leq j\leq d+1}$. Each $x_{j}=\pm1$ so each entry (except for the very last one) differs from the others by only a sign. So in essence only one multiplication is necessary to obtain each vector $y_{j}$.

	To get the frame vectors each vector $y_{j}$ has to be scaled. The appropriate scaling factors for each vector are the constants
	\[
		c_{j} = \sqrt{\frac{d+1}{d}}\frac{1}{\|y_{j}\|} = \sqrt{\frac{d+1}{d}}\sqrt{\frac{j}{j+1}}.
	\]
	The matrix that gives the associated frame is the matrix $V = [{v}_{1}\ldots{v}_{d}]$ where each vector ${v}_{j}$ is given by
	\[
		{v}_{j} = c_{j}y_{j}.
	\]
	Since every entry of $y_{j}$ (except for the $(j+1)^{\text{th}}$ entry) differs from the others by only a sign, only two multiplications (one for the first $j$ entries and one for the $(j+1)^{\text{th}}$ entry) are essentially necessary to obtain ${v}_{j}$ from $y_{j}$. So with these assumptions it appears that to get the frame vectors from the given vector $x$ requires $2(d+1)$ multiplications.
\end{remark}

%-----------------------------------------
\section{Construction of $k$-angle tight frames}\label{k_angle_tight_frames}
\subsection{$\boldsymbol{2}$-angle tight frames}
\label{2_3_dist}

As a first step towards generalizing ETFs, one considers constructing $2$-angle tight frames. In Example~\ref{example::RealHadamardConstruction} below, several examples of $2$-angle tight frames are presented. The following lemma is needed.

\begin{lemma}\label{lemma::UnitaryConstruction}
	Let $d\in\mathbb{N}$ and let $J$ denote the $d\times d$ matrix whose entries are all one. Then the matrix $U$ given by $U = \frac{2}{d} J -I_d$ is orthogonal, where $I_d$ is the $d \times d$ identity matrix.
\end{lemma}
\begin{proof}
	Since $J^{2} = dJ$, and $\frac 2d J - I_d$ is symmetric, it follows that
	\[
		 \left(\frac{2}{d}J-I_d\right)\left(\frac{2}{d}J-I_d\right)^{\mathrm{T}} = \left(\frac{2}{d}J-I_d\right)\left(\frac{2}{d}J-I_d\right) =  \frac{4}{d^{2}}J^{2} - \frac{4}{d}J + I_d = I_d.
	\]
\end{proof}

%The matrix $\frac{2}{d}J$ can be used to obtain real $2$- or $3$-angle tight frames of $2d$ vectors. A real $d\times d$ \textit{Hadamard matrix} may also be used if one is available (see Remark~\ref{remark::ExistenceOfHadamards}).

\begin{definition}
	\rm
	A $d\times d$ matrix $H$ is said to be a \textbf{real Hadamard matrix} if $HH^{\mathrm{T}} = dI_{d}$ and the entries of $H$ are either $-1$ or $1$. Similarly, $H$ is said to be a \textbf{complex Hadamard matrix} if $HH^{*} = dI_{d}$ and the entries of $H$ are unimodular.
\end{definition}

If $H$ is a $d\times d$ real (respectively, complex) Hadamard matrix, then $\frac{1}{\sqrt{d}}H$ is orthogonal (respectively, unitary).

\begin{remark}\label{remark::ExistenceOfHadamards}
	\rm
	The existence and classification of real and complex Hadamard matrices is an important open problem, although the complex case provides more options. In particular, a $d\times d$ complex Hadamard matrix for any $d\in\mathbb{N}$ is given by the DFT matrix with unimodular entries. Real Hadamard matrices are rarer, but a construction due to Sylvester provides a $2^{n}\times 2^{n}$ Hadamard matrix for every $n\in\mathbb{N}$~\cite{Sylvester1867}.
\end{remark}

\begin{example}\label{example::RealHadamardConstruction}
\rm
 	Let $\mathcal{F}_{1}$ be the standard basis of $\mathbb{R}^{d}$ or $\mathbb{C}^{d}$. In each example below, the tightness of the resulting frame follows from the fact that the union of two finite unit-normed tight frames of a vector space is again a finite unit-normed tight frame for the same vector space.
 	\begin{enumerate}[i.]
 	\item
 	Let $\mathcal{F}_{2}$ be the orthonormal basis of $\mathbb{R}^{d}$ obtained from the columns of the matrix $U$ in Lemma~\ref{lemma::UnitaryConstruction}. If $d = 4$ then $\mathcal{F}_{1}\cup\mathcal{F}_{2}$ is a real $(8,4)$ $2$-angle tight frame, otherwise, $\mathcal{F}_{1}\cup\mathcal{F}_{2}$ is a real $(2d,d)$ $3$-angle tight frame.

 The Gram matrix of $\mathcal{F}_{1}\cup\mathcal{F}_{2}$ is
		\[
			G_{1} = F_{1}^{\mathrm{T}}F_{1} =
			\begin{bmatrix}
				I_{d} & U \\
				U^{\mathrm{T}} & I_{d} \\	
			\end{bmatrix}
			=
			\begin{bmatrix}
				I_{d} & \frac2d J - I_{d} \\
				\frac2d J-I_{d} & I_{d}
			\end{bmatrix}.
		\]
		The only possible moduli of the off-diagonal entries in $G_{1}$ are $0, \frac{2}{d},$ and $1-\frac{2}{d}.$ When $d = 4,$ the only possible moduli are $0$ and $\frac{1}{2}.$

 	\item
 	Suppose that a real $d\times d$ Hadamard matrix $H$ exists and let $\mathcal{F}_{3}$ be the orthonormal basis of $\mathbb{R}^{d}$ obtained from the columns of $\frac{1}{\sqrt{d}}H$. Then $\mathcal{F}_{1}\cup\mathcal{F}_{3}$ is a real $(2d,d)$ $2$-angle tight frame. The only possible moduli of the off-diagonal entries in the Gram matrix are $0$ and $\frac{1}{\sqrt{d}}.$

    \item
    Let $\mathcal{F}_{4}$ be the orthonormal basis of $\mathbb{C}^{d}$ obtained from the columns of the normalized DFT matrix. Then $\mathcal{F}_{1}\cup\mathcal{F}_{4}$ is a complex $(2d,d)$ $2$-angle tight frame. Again, the moduli of the off-diagonal entries in the Gram matrix are either $0$ or $\frac{1}{\sqrt{d}}$.
 	\end{enumerate}
\end{example}

The construction in Example~\ref{example::RealHadamardConstruction} iii. will also provide $(2d,d)$ $2$-angle tight frames if the normalized DFT matrix is replaced by an arbitrary normalized complex Hadamard matrix as shown in Theorem \ref{theorem::MutuallyUnbiasedHadamardsConstruction}. Going further, \textit{mutually unbiased Hadamards} can be used to construct $2$-angle tight frames with higher redundancy.

\begin{definition}
	\rm
	Consider a collection $\{H_{1}, H_{2}, \ldots, H_{n}\}$ of $d\times d$ Hadamard matrices. These matrices are said to be \textbf{mutually unbiased Hadamards} if $\frac{1}{\sqrt{d}}H_{j}^{*}H_{k}$ is again a Hadamard matrix for all $1\leq j < k\leq n$.
\end{definition}

As mentioned in~\cite{Durt2010}, the construction of $n$ mutually unbiased Hadamards of size $d\times d$ is equivalent to the construction of $n+1$ \textit{mutually unbiased bases} (MUBs); that is, a collection $\{\mathcal{E}_{1},\ldots,\mathcal{E}_{n+1}\}$ of orthonormal bases $\mathcal{E}_{j} = \{{e}_{l}^{(j)}\}_{l=1}^{d}$ such that $|\langle{e}_{l}^{(j)},{e}_{m}^{(k)}\rangle| = \frac{1}{\sqrt{d}}$ for $1\leq l,m\leq d$ and $1\leq j<k\leq n+1$. It is known from \cite{KR1} that the maximal set of MUBs in any given $d$-dimensional Hilbert space is of size at most $d+1$. Constructions presented in~\cite{KR1} provide MUBs of maximal size (that is, $d+1$ MUBs in a $d$-dimensional space) in any space whose dimension is $p^q$ for prime $p.$ The question of the existence of maximal MUBs in other dimensions remains an open problem.

\begin{theorem}\label{theorem::MutuallyUnbiasedHadamardsConstruction}
	Let $d,n\in\mathbb{N}$.
	\begin{enumerate}[i.]
		\item
		Let $H$ be a $d\times d$ Hadamard matrix. Then the columns of
		\[
			\begin{bmatrix}
				I_{d} & \frac{1}{\sqrt{d}}H	
			\end{bmatrix}
		\]
		form a $2$-angle $(2d,d)$ tight frame.

		\item
		Let $\{H_{1}, H_{2}, \ldots, H_{n}\}$ be a collection of $d\times d$ mutually unbiased Hadamards where $n\leq d$. Then the columns of
		\[
			\begin{bmatrix}
				I_{d} &	\frac{1}{\sqrt{d}}H_{1} & \frac{1}{\sqrt{d}}H_{2} & \cdots & \frac{1}{\sqrt{d}}H_{n}
			\end{bmatrix}
		\]
		form a $2$-angle $((n+1)d,d)$ tight frame.
	\end{enumerate}
\end{theorem}
\begin{proof}
$\,$
	\begin{enumerate}[i.]
		\item
		The justification of this statement is the same as the one given in Example~\ref{example::RealHadamardConstruction} part iii. Just replace the DFT matrix with $\frac{1}{\sqrt{d}}H$.

		\item
		The frame is a union of $n+1$ orthonormal bases and so must be a tight frame. It remains to show that the frame is a $2$-angle frame. Let
		\[	F_2 =
			\begin{bmatrix}
				I_{d} & \frac{1}{\sqrt{d}}H_{1} & \cdots & \frac{1}{\sqrt{d}}H_{n}	
			\end{bmatrix}.
		\]
		The Gram matrix $G_{2}$ of this frame is
		\[
			G_{2} = F_{2}^{*}F_{2} =
			\begin{bmatrix}
				I_{d} & \frac{1}{\sqrt{d}}H_{1} & \frac{1}{\sqrt{d}}H_{2} & \cdots & \frac{1}{\sqrt{d}}H_{n} \\
				\frac{1}{\sqrt{d}}H_{1}^{*} & I_{d} & \frac{1}{d}H_{1}^{*}H_{2} & \cdots & \frac{1}{d}H_{1}^{*}H_{n} \\
				\vdots & \vdots & \vdots & \ddots & \vdots \\
				\frac{1}{\sqrt{d}}H_{n}^{*} & \frac{1}{d}H_{n}^{*}H_{1} & \frac{1}{d}H_{n}^{*}H_{2} & \cdots & I_{d}
			\end{bmatrix}.
		\]
		Since $\{H_{1},\ldots,H_{n}\}$ is a collection of mutually unbiased Hadamards, each entry in $\frac{1}{d}H_{j}^{*}H_{k}$ for $1\leq j<k\leq n$ has modulus $\frac{1}{\sqrt{d}}$, as does each entry in $\frac{1}{\sqrt{d}}H_{j}$ for $1\leq j\leq n$ . Therefore each off-diagonal entry of $G_{2}$ has modulus either $0$ or $\frac{1}{\sqrt{d}}$, which implies that the frame is a $2$-angle frame.
	\end{enumerate}
\end{proof}

%\begin{remark}
%	\rm
%	Constructions presented in~\cite{KR1} provide MUBs of maximal size (that is, $d+1$ MUBs in a $d$-dimensional space) in any space whose dimension is $p^q$ for prime $p.$ The question of the existence of maximal MUBs in other dimensions remains an open problem.
%\end{remark}
%-----------------------------------------
\subsection{Construction of $\boldsymbol{k}$-angle tight frames; $\boldsymbol{k \geq 2}$}
\label{k_dist_tight}

In this subsection, a general method of constructing tight frames is presented such that for a given $k,$ the number of distinct angles between vectors is at most $k.$

\begin{theorem}\label{theorem::binomialConstruction}
	Let $d,k\in\mathbb{N}$ with $k< d+1$, and set $d' = \binom{d+1}{k}$. Denote the collection of all subsets of $\{1,\ldots,d+1\}$ of size $k$ by $\{\Lambda_{i}\}_{i=1}^{d'}$. Let $\{f_{i}\}_{i=1}^{d+1}\subseteq\mathbb{R}^{d}$ denote the ETF with $\dotProduct{f_{i}}{f_{j}} = -\frac{1}{d}$ for all $i\neq j$. Define a new collection $\{g_{i}\}_{i=1}^{d'}$ as follows:
	\[
		g_{i} := \frac{\sum_{j\in\Lambda_{i}}f_{j}}{\|\sum_{j\in\Lambda_{i}}f_{j}\|}.
	\]
	Then $\{g_{i}\}_{i=1}^{d'}$ forms a $\hat{k}$-angle tight frame of $d' $ vectors in $\mathbb{R}^{d}$, where $\hat{k}\leq k$.
\end{theorem}

To prove this theorem, the following results are needed.

\begin{lemma}\label{lemma::sumNorms}
%	Let $\{f_{i}\}_{i=1}^{d+1}$ denote the ETF given in the hypothesis of Theorem~\ref{theorem::binomialConstruction}. Then $\| \sum_{j\in\Lambda_{i}}f_{j} \|$ is independent of $i$.
Under the setting and assumptions of Theorem~\ref{theorem::binomialConstruction}, $\| \sum_{j\in\Lambda_{i}}f_{j} \|$ is independent of $i$.
\end{lemma}
\begin{proof}
	By a direct calculation,
	\begin{align*}
		\left\|\sum_{j\in\Lambda_{i}}f_{j}\right\|^{2}
		&= \dotProduct{\sum_{j\in\Lambda_{i}}f_{j}}{\sum_{j'\in\Lambda_{i}}f_{j'}} \\
		&= \sum_{j\in\Lambda_{i}}\sum_{j'\in\Lambda_{i}}\dotProduct{f_{j}}{f_{j'}} \\
		&= \sum_{j\in\Lambda_{i}}\|f_{j}\|^{2} + \sum_{j\neq j'}\dotProduct{f_{j}}{f_{j'}}.
	\end{align*}
	The right hand side simplifies to $k+k(k-1)(-\frac{1}{d})$, and so for all $i$
	\[
		\left\|\sum_{j\in\Lambda_{i}}f_{j}\right\| = \sqrt{\frac{k(d+1-k)}{d}}.
	\]
\end{proof}

\begin{lemma}\label{lemma::binaryMatrix}
	Let $K$ denote the matrix whose columns are the binary vectors in $\mathbb{R}^{d+1}$ with exactly $k$ ones and note that there are $d' = \binom{d+1}{k}$ such vectors. In particular, set
	\[
		K = \begin{bmatrix}
			k_{1} & \cdots & k_{d'}	
		\end{bmatrix}
	\]
	where $\operatorname{supp}k_{j} = \Lambda_{j}$. Then
	\[
		KK^{\mathrm{T}} = \binom{d-1}{k-1}I_{d+1}+\binom{d-1}{k-2}J
	\]
	where $J$ is the $(d+1)\times(d+1)$ matrix of ones.
\end{lemma}
\begin{proof}
	Set $K = [k_{ij}]$ for $1\leq i\leq d+1$ and $1\leq j \leq d'$ and note that $k_{ij} = 1$ if and only if $i\in\Lambda_{j}$. Let
	\[
		\tilde{k}_{ij} = \sum_{m=1}^{d'}k_{im}k_{jm}
	\]
	denote the $(i,j)^{\mathrm{th}}$ entry of $KK^{\mathrm{T}}$. Then $\tilde{k}_{ii} = \sum_{m=1}^{d'}k_{im}^{2}$ is precisely the number of subsets $\Lambda_{m}\subseteq\{1,\ldots,d+1\}$ of size $k$ that contain $i$, so $\tilde{k}_{ii} = \binom{d}{k-1} = \binom{d-1}{k-1}+\binom{d-1}{k-2}$. Similarly, if $i\neq j$ then $\tilde{k}_{ij} = \sum_{m=1}^{d'} k_{im}k_{jm}$ counts the number of subsets $\Lambda_{m}$ that contain both $i$ and $j$, so $\tilde{k}_{ij} = \binom{d-1}{k-2}$ if $i\neq j$. Thus $KK^{\mathrm{T}}$ has the desired form.
\end{proof}
The \emph{frame potential} \cite{BF03} of a set of vectors $\{x_i\}_{i=1}^N$ is
$$FP\{x_i\}_{i = 1}^N := \sum_{i=1}^N \sum_{j=1}^N |\langle x_i, x_j\rangle|^2.$$
Note that the frame potential is the trace of the square of the Gram matrix of $\{x_i\}_{i=1}^N.$

\begin{theorem}\label{BenedettoFickusFramePotential}\cite{BF03}
For a set of $N$ unit vectors $\{x_i\}_{i=1}^N$ in a $d$-dimensional space, if $N \geq d,$ the minimum value of the frame potential is $N^2/d,$ and the minimizers are precisely the unit normed tight frames of the underlying space.
\end{theorem}
The proof of Theorem~\ref{theorem::binomialConstruction} is now provided below.
\begin{proof}[Proof of Theorem~\ref{theorem::binomialConstruction}]
First it will be shown that $\{g_{i}\}_{i=1}^{d'}$ as defined in the statement of the theorem is a unit normed tight frame.
	Let $K$ denote the matrix given in Lemma~\ref{lemma::binaryMatrix}. If $F$ is the matrix with columns $\{f_{i}\}_{i=1}^{d+1}$, then it follows that $Fk_{i} = \sum_{j\in\Lambda_{i}}f_{j}$. The matrix with columns $\{g_{i}\}_{i=1}^{d'}$ can then be written as
	\[
		\sqrt{\frac{d}{k(d+1-k)}}FK,
	\]
	where the scalar term comes from Lemma~\ref{lemma::sumNorms}. This implies that the Gram matrix $G_{1}$ of $\{g_{i}\}_{i=1}^{d'}$ is the matrix
	\[
		\frac{d}{k(d+1-k)}(FK)^{\mathrm{T}}(FK) = \frac{d}{k(d+1-k)}K^{\mathrm{T}}GK
	\]
where $G$ denotes the Gram matrix of $\{f_{i}\}_{i=1}^{d+1}$. It will be shown that $\{g_{i}\}_{i=1}^{d'}$ is tight by computing its frame potential and using Theorem~\ref{BenedettoFickusFramePotential}. Let $c_{1} = \binom{d-1}{k-1}$ and $c_{2} = \binom{d-1}{k-2}.$ Then by Lemma~\ref{lemma::binaryMatrix}
	\begin{align*}
		FP\{g_{i}\}_{i=1}^{d'}
		&= \operatorname{tr}G_{1}^{2} \\
		& = \left(\frac{d}{k(d+1-k)}\right)^{2}\operatorname{tr}(K^{\mathrm{T}}GKK^{\mathrm{T}}GK) \\
		&= \left(\frac{d}{k(d+1-k)}\right)^{2}\operatorname{tr}(K^{\mathrm{T}}G(c_{1}I+c_{2}J)GK).
	\end{align*}
	According to the hypothesis of Theorem~\ref{theorem::binomialConstruction}, $\dotProduct{f_{i}}{f_{j}} = -\frac{1}{d}$ 	for all $i\neq j.$ This makes the product $GJ$ equal to the $(d+1)\times(d+1)$ zero matrix. Therefore,
	%Furthermore, since $GJ$ is the $(d+1)\times(d+1)$ zero matrix,
	\begin{align*}
		FP\{g_{i}\}_{i=1}^{d'}
		&= c_{1}\left(\frac{d}{k(d+1-k)}\right)^{2}\operatorname{tr}(K^{\mathrm{T}}G^{2}K) \\
		&= c_{1}\left(\frac{d}{k(d+1-k)}\right)^{2}\operatorname{tr}(G^{2}KK^{\mathrm{T}}) \\
		&= c_{1}\left(\frac{d}{k(d+1-k)}\right)^{2}\operatorname{tr}(G^{2}(c_{1}I+c_{2}J)) \\
		&= c_{1}^{2}\left(\frac{d}{k(d+1-k)}\right)^{2}\operatorname{tr}(G^{2}) \\
		&= \left(\frac{d}{k(d+1-k)}c_{1}\right)^{2}\frac{(d+1)^{2}}{d}
	\end{align*}
	where the last equality follows from the fact that $\{f_{i}\}_{i=1}^{d+1}$ is a unit normed tight frame and the result in Theorem~\ref{BenedettoFickusFramePotential}. Further simplification gives
	\begin{align*}
		FP\{g_{i}\}_{i=1}^{d'}
		&= \left(\frac{d}{k(d+1-k)}c_{1}\right)^{2}\frac{(d+1)^{2}}{d} \\
		&= \left[\frac{(d+1)d}{k(d+1-k)}\binom{d-1}{k-1}\right]^{2}\frac{1}{d} \\
		&= \binom{d+1}{k}^{2}\frac{1}{d} \\
		&= \frac{(d')^{2}}{d}.
	\end{align*}
	Hence $\{g_{i}\}_{i=1}^{d'}$ is a unit normed tight frame for $\mathbb{R}^{d}$ by Theorem~\ref{BenedettoFickusFramePotential}.

	It remains to be shown that the set $\{g_{i}\}_{i=1}^{d'}$ is also a $\hat{k}$-angle frame where $\hat{k}\leq k$.
	Let $i,j\leq d'$ with $i\neq j$.
	By the proof of Lemma~\ref{lemma::sumNorms}
	\begin{align*}
		\dotProduct{g_{i}}{g_{j}}
		%&= \left\| \sum_{j\in\Lambda_{i}}f_{j} \right\|^{-2} \dotProduct{\sum_{i'\in\Lambda_{i}}f_{i'}}{\sum_{j'\in\Lambda_{j}}f_{j'}} \\ \displaybreak
		&= \frac{d}{k(d+1-k)}\sum_{i'\in\Lambda_{i}}\sum_{j'\in\Lambda_{j}}\dotProduct{f_{i'}}{f_{j'}}.
	\end{align*}
	 Now set $l = |\Lambda_{i}\cap\Lambda_{j}|$. Then the double summation can be rewritten as
	\begin{align*}
		 \sum_{i'\in\Lambda_{i}}\sum_{j'\in\Lambda_{j}}\dotProduct{f_{i'}}{f_{j'}}
		&= \sum_{i'\in\Lambda_{i}\cap\Lambda_{j}}\sum_{j'\in\Lambda_{j}\cap\Lambda_{i}}\dotProduct{f_{i'}}{f_{j'}}
		 + \sum_{i'\in\Lambda_{i}\setminus\Lambda_{j}}\sum_{j'\in\Lambda_{j}\cap\Lambda_{i}}\dotProduct{f_{i'}}{f_{j'}} \\
		&\quad{} + \sum_{i'\in\Lambda_{i}\cap\Lambda_{j}}\sum_{j'\in\Lambda_{j}\setminus\Lambda_{i}}\dotProduct{f_{i'}}{f_{j'}} + \sum_{i'\in\Lambda_{i}\setminus\Lambda_{j}}\sum_{j'\in\Lambda_{j}\setminus\Lambda_{i}}\dotProduct{f_{i'}}{f_{j'}} \\
		&= \left[l(1) - l(l-1)\frac{1}{d} \right] + \left[-\frac{1}{d}(k-l)l \right] + \left[-\frac{1}{d}(k-l)l \right] + \left[-\frac{1}{d}(k-l)^{2}\right] \\
		&= l-\frac{1}{d}(k^{2}-l).
	\end{align*}
	Therefore
	\[
		\dotProduct{g_{i}}{g_{j}} = \frac{d}{k(d+1-k)}\left[l-\frac{1}{d}(k^{2}-l)\right] = \frac{l(d+1)-k^{2}}{k(d+1-k)}.
	\]
	Since $0\leq l \leq k-1$ if $i\neq j$, there are $k$ different choices for $l$ in the above formula. Hence $\dotProduct{g_{i}}{g_{j}}$ can take on at most $k$ different values when $i\neq j$, which finishes the proof.
\end{proof}

%-----------------------------------------
\section{Acknowledgments}
\label{Acknowledge}
The authors would like to acknowledge support from the National Science Foundation under Award No. CCF-1422252. The authors are grateful to Doug Cochran and Kasso Okoudjou for their valuable insight and suggestions on this subject.

%%\bibliographystyle{elsarticle-num}
%%\bibliographystyle{elsarticle-harv}
%\nocite{*}
%\section*{References}
%\bibliographystyle{plain} %sort alphabetically by last name
\bibliographystyle{unsrt} %sort in numerical order of appearance
\bibliography{ETF_refs}
\end{document}